\PassOptionsToPackage{prologue,usenames,dvipsnames,table}{xcolor}
\PassOptionsToPackage{bookmarks,unicode,colorlinks=true}{hyperref}
\newif\ifanonymous

\anonymousfalse

\documentclass[runningheads]{llncs}

\usepackage{amssymb,amsmath,mathtools}
\usepackage{xspace}
\usepackage{thm-restate}
\usepackage[hyperfootnotes=true,colorlinks,linkcolor=RedViolet,anchorcolor=RedViolet,citecolor=blue,urlcolor=black]{hyperref}
\usepackage[capitalize,nameinlink]{cleveref}
\usepackage{minted, tabularray}
\usepackage{float}  
\usepackage[inline]{enumitem}
\usepackage{bm}
\usepackage{nicefrac}
\usepackage{adjustbox}
\usepackage{wrapfig}
\usepackage{orcidlink}
\usepackage[utf8]{inputenc}
\usepackage{caption,subcaption}      
\usepackage{tcolorbox}               
\tcbuselibrary{listings,skins}       
\usepackage{listings}                
\usepackage{xcolor}
\usepackage{fontawesome5}
\usepackage{booktabs,tabularx}
\usepackage{threeparttable}
\usepackage{pifont}
\usepackage{siunitx}
\usepackage{listings} 
\definecolor{codegreen}{rgb}{0,0.6,0}
\definecolor{codegray}{rgb}{0.5,0.5,0.5}
\definecolor{codepurple}{rgb}{0.58,0,0.82}
\definecolor{backcolour}{rgb}{0.9647, 0.9725, 0.9804}
\definecolor{codered}{rgb}{0.698,0.133,0.133}
\definecolor{codeblue}{rgb}{0.169, 0.431, 0.639}
\lstdefinestyle{mystyle}{
language=Python,
backgroundcolor=\color{backcolour}, 
commentstyle=\color{codegreen},
numberstyle=\tiny\color{codegray},
stringstyle=\color{codepurple},
basicstyle=\ttfamily\scriptsize,
breakatwhitespace=false, 
breaklines=true, 
captionpos=b, 
numbers=left, 
numbersep=5pt, 
showspaces=false, 
showstringspaces=false,
showtabs=false, 
tabsize=1,
xleftmargin=0.3cm,
mathescape,
keywordstyle=\color{codegray}\bfseries,
morekeywords={from,import}
}

\crefname{figure}{Fig.}{Fig.}
\crefname{lemma}{Lem.}{Lem.}
\crefname{example}{Exmp.}{Exmp.}
\crefname{section}{Sect.}{Sect.}
\crefname{problem}{Prob.}{Prob.}
\crefname{appendix}{Appx.}{Appx.}
\crefname{definition}{Def.}{Def.}
\crefname{theorem}{Thm.}{Thm.}
\crefname{corollary}{Cor.}{Cor.}
\crefname{algorithm}{Alg.}{Alg.}
\crefname{natation}{Not.}{Not.}
\input{macros.tex}

\newcommand{\rev}[1]{\textcolor{black}{#1}}

\makeatletter
\def\thanks#1{\protected@xdef\@thanks{\@thanks
        \protect\footnotetext{#1}}}
\makeatother

\begin{document}
\title{
STLts-Div: Diversified Trace Synthesis from STL Specifications Using MILP\\
(Extended Version)
\thanks{The authors are supported by JST ASPIRE Grant No.\ JPMJAP2301. \rev{Jie An is supported by the CAS Project for Young Scientists in Basic Research Grant No.\ YSBR-123. Zhenya Zhang is supported by JST BOOST Grant No. JPMJBY24D7 and JSPS Grant No. JP25K21179.} The work was done during M.J.'s internship at National Inst.\ of Inform., Japan.}
}

\titlerunning{STLts-Div: Diversified Trace Synthesis from STL Using MILP}
%

\ifanonymous
    \author{}
    \institute{}
\else
    \author{Martin Jouve-Genty\inst{1,2}\orcidlink{0009-0008-9912-3225}\textsuperscript{$\dagger$}
    \and
    Han Su\inst{2}\orcidlink{0000-0003-4260-8340}\textsuperscript{$\dagger$}
    \and
    Sota Sato\inst{2,5}\orcidlink{0000-0001-7147-3989}
    \and
    Jie An\inst{3}\orcidlink{0000-0001-9260-9697}
    \and 
    Zhenya Zhang\inst{4,2}\orcidlink{0000-0002-3854-9846}
    \and
    Ichiro Hasuo\inst{2,5,6}\orcidlink{0000-0002-8300-4650}
    }
    \authorrunning{M.~Jouve-Genty et~al.}

    \institute{Département Informatique, ENS de Lyon, Lyon, France \\
    \email{martin.jouve-genty@ens-lyon.fr}
    \and 
    National Institute of Informatics, Tokyo, Japan \\
    \email{\{suhan,sotasato,hasuo\}@nii.ac.jp}
    \and
    Institute of Software, Chinese Academy of Sciences, Beijing, China \\
    \email{anjie@iscas.ac.cn}\and
    Kyushu University, Fukuoka, Japan \\
    \email{zhang@ait.kyushu-u.ac.jp}\and
    SOKENDAI (The Graduate University for Advanced Studies), Tokyo, Japan\and
    Imiron Co., Ltd., Tokyo, Japan
    }
\fi
\maketitle              
\def\thefootnote{$\dagger$}
\footnotetext{Equal contribution.}
\renewcommand{\thefootnote}{\arabic{footnote}}

\begin{abstract}
    Modern cyber–physical systems are complex, and requirements are often written in Signal Temporal Logic (STL). Writing the right STL is difficult in practice; engineers benefit from concrete executions that illustrate what a specification actually admits. Trace synthesis addresses this need, but a single witness rarely suffices to understand intent or explore edge cases---diverse satisfying behaviors are far more informative. We introduce diversified trace synthesis: the automatic generation of sets of behaviorally diverse traces that satisfy a given STL formula. Building on a MILP encoding of STL and system model, we formalize three complementary diversification objectives---Boolean distance, random Boolean distance, and value distance---all captured by an objective function and solved iteratively. We implement these ideas in STLts-Div, a lightweight Python tool that integrates with Gurobi.
\end{abstract}
\keywords{Cyber-Physical Systems  \and Signal Temporal Logic \and Trace Synthesis. }

\section{Introduction}

    \emph{Cyber–physical systems} (CPS) underpin safety-critical applications such as aircraft collision-avoidance protocols \cite{tomlin2000game}, aerospace \cite{atkins2013aerospace}, and pacemakers \cite{ye2008modelling}. Ensuring their reliability is essential, yet formal verification is challenging due to hybrid/nonlinear dynamics, continuous state spaces, and environmental uncertainty. As a result, full verification is often infeasible in practice.

    \emph{Trace synthesis} offers a pragmatic alternative: instead of proving properties for all behaviors, it constructs concrete executions that satisfy (or violate) a given temporal specification, providing actionable artifacts for testing and analysis. As a dual to model checking \cite{Sato2024}, trace synthesis has been realized via SMT encodings \cite{bae2019bounded}, optimization-based falsification \cite{ernst2021arch}, and—--more recently--—mixed-integer linear programming (MILP) encodings that capture both \emph{Signal Temporal Logic (STL)} constraints and plant dynamics \cite{Sato2024}. These approaches generate informative scenarios and counterexamples while interfacing naturally with simulators and models.

    However, a single satisfying trace is often insufficient—for example, when validating whether an STL requirement truly captures the intent of the engineer. A \emph{set} of behaviorally diverse traces that illuminate different ways the specification can be met is required. In this work we study \emph{diversified trace synthesis} from STL, aiming to generate traces that satisfy the specification while exhibiting sufficient diversity.
    

    We address this problem with STLts-Div, a tool that synthesizes sets of diverse traces for a given system model \(\model\) and STL specification \(\phi\). STLts-Div builds on the MILP-based synthesis framework of \cite{Sato2024}. We introduce three complementary diversity objectives---Boolean distance (BD), randomized Boolean distance (RBD), and value distance (VD)---and encode them as MILP objectives. We adopt the MILP setting because its numeric decision variables allow these objectives to be integrated directly, whereas propositional SMT encodings do not readily support summing/counting of truth assignments (see Remark~\ref{rmk:MILP-good}).

    \myparagraph{Contributions and Organization}
    We summarize our contributions as below.
    \begin{itemize}
        \item We formulate the problem of diversified trace synthesis for STL. 
        \item We propose three diversity objectives (BD, RBD, VD) and provide MILP objective function encodings compatible with \cite{Sato2024}.
        \item We implement these ideas in STLts-Div, a Python package, and evaluate them on benchmark suites, demonstrating substantially higher diversity than a solution-pool baseline.
    \end{itemize}

    \Cref{sec:pre} reviews STL preliminaries and states the problem. \Cref{sec:milp} recaps the MILP encoding of \cite{Sato2024}, which underpins our approach. \Cref{sec:diver} formalizes three diversity objectives and presents our MILP-based synthesis procedures. \Cref{sec:demo} demonstrates the STLts-Div tool. \Cref{sec:exp} reports experiments evaluating effectiveness and efficiency. We conclude in \Cref{sec:conclusion}.

    \myparagraph{Related Work}
    Signal Temporal Logic is widely used for specifying CPS behaviors. Prior work spans controller synthesis \cite{donze2015blustl,raman2015reactive,raman2014model,su2024switching}, runtime/online monitoring \cite{su2025runtime,zhang2023online}, and trace synthesis/falsification \cite{Sato2024,bae2019bounded}, among others. The choice of technique depends on the task: for continuous-time control, Lyapunov/barrier–based methods synthesize controllers that enforce STL constraints \cite{lindemann2018control}; for bounded model checking, SMT encodings search for a time partition witnessing satisfaction \cite{bae2019bounded}; and recent MILP encodings support trace synthesis by jointly searching over signal values and a \(\delta\)-stable partition (see \cref{sec:milp}) \cite{Sato2024}. Existing approaches, however, typically return one satisfying execution. In contrast, we target the multi-trace setting and synthesize sets of behaviorally diverse traces for a fixed model and STL formula.

\section{Preliminaries}\label{sec:pre}

    Let $\Nats$, $\Reals$, and $\NonNegReals$ denote the sets of natural numbers, real numbers, and non-negative real numbers, respectively. Let $\Bool = \{\top, \bot\}$ denote the Boolean domain. For a set \(X\), we write \(\pow{X}\) for its power set and \(\card{X}\) for its cardinality. An \emph{interval} is any subset of \(\NonNegReals\) of the form \((a,b), (a,b], [a,b),\) or \([a,c]\), where \(a < b\) and \(a \le c\).

    \myparagraph{Signal}
    Let \(T>0\) and \(V\) be a finite set of variables. A \emph{signal} over \(V\) with time horizon \(T\) is a function \(\sigma \colon [0,T] \to \Reals^{\card{V}}\). We let \(\signal_V^T\) denote the set of all such signals; when \(T\) is clear from the context we write \(\signal_V\). For \(v\in V\) we write  $\sigma(t)(v)$ for the value of the variable $v$ in $\sigma$ at time $t$. We extend any \(\sigma\in\signal_V^T\) to \(t > T\) by setting \(\sigma(t)\ddef\sigma(T)\) for all \(t > T\).

    For a signal \(\sigma\) and \(t\ge 0\), the \(\tpost\) is the signal \(\sigma^t\) defined by \(\sigma^t(t^\prime) \ddef\sigma(t+t^\prime) \). The \(\tpost\) is the basis for defining the STL semantics (see~\cref{def:stl-sem}). 

    \myparagraph{System model}
    Let \(V^\prime\) and \(V\) be the input and output variable sets. A \emph{system model} \(\model\) from \(V^\prime\) to \(V\) with a time horizon \(T\) is a (possibly nondeterministic) map \(\model \colon \signal_{V^\prime}^T \to \pow{\signal_{V}^T}\).
    The trace set of \(\model\) is \(\mathcal{L}(\model)\ddef\bigcup_{\tau\in\signal_{V^\prime}^T}\model(\tau)\), \ie the set of all output signals under any input $\tau$.

    
    \begin{example}[Single-lane car]
        Consider a car moving in a single lane whose dynamics are given by the ordinary differential equation (ODE)
        \begin{align}\label{eq:ode}
            \dot{x}=v,\qquad \dot{v}=a,
        \end{align}
        where \(x\), \(v\) and \(a\) denote position, velocity and acceleration, respectively.

        We formalize this system as \(\model_{\mathrm{car}}\) with input variables \(V'=\{a, v^{\mathrm{init}}, x^{\mathrm{init}}\}\) and output variables \(V=\{a, v, x\}\). Here, the acceleration $a$ changes over time; the initial velocity and position $v^{\mathrm{init}}, x^{\mathrm{init}}$ change as well but only their initial values $\tau(0)(v^{\mathrm{init}}), \tau(0)(x^{\mathrm{init}})$ matter. Given an input signal \(\tau\), the model yields \(\model_{\mathrm{car}}(\tau)=\{\sigma\}\), where \(\sigma\) is determined by \cref{eq:ode}: \(\sigma(t)(a) = \tau(t)(a)\), \(\sigma(t)(v) = \tau(0)(v^{\mathrm{init}}) + \int_0^t \tau(\theta)(a)\,\dif\theta\), \(\sigma(t)(x) = \tau(0)(x^{\mathrm{init}}) + \int_0^t \sigma(\theta)(v)\,\dif\theta\). 
    \end{example}

    \subsection{Signal Temporal Logic (STL)}\label{subsec:stl}

        We use the standard syntax and semantics of STL~\cite{Maler2004,Sato2024}. We present some basic definitions for the record.

        In STL, an \emph{atomic proposition} over a variable set \(V\) is represented as \( p\lddef (\pi_p(\vec{w}) \ge 0) \), where \(\pi_p:\Reals^{\lvert V\rvert} \to \Reals\) maps a valuation \(\vec{w}\in\Reals^{\card{V}}\) of variables in \(V\) to a real value. In this work, we restrict attention to \emph{linear} (affine) atomic propositions, \ie \(\pi_p(\vec{w})=a^\top\vec{w}+c\) with \(a\in\Reals^{\card{V}}\) and \(c\in\Reals\). This syntactic restriction enables an MILP encoding of STL formulas.

        Any STL formula has an equivalent \emph{negation normal form} (NNF)~\cite{Fainekos2009}, in which negations appear only adjacent to atomic propositions. We assume all STL formulas are in NNF in this work. The syntax of STL is as follows:
        \begin{align*}
            \phi \lddef \top \mid p \mid \neg p \mid \phi_1\wedge\phi_2 \mid \phi_1\vee\phi_2 \mid \phi_1\unt_I\phi_2 \mid \phi_1 \rels_I \phi_2,
        \end{align*}
        where $I$ is a nonsingular closed time interval, and \(\unt_I\) and \(\rels_I\) are temporal operators \emph{until} and \emph{release}, respectively. The \emph{eventually} and \emph{always} operators are syntactic sugars defined as \(\ev_I\phi=\top\unt_I\phi\), \(\alw_I\phi=\bot\rels_I\phi\). The set of all subformulas of \(\phi\) is denoted by \(\sub\).

        The Boolean semantics \(\sigma\vDash\phi\) and robust semantics \(\sem{\sigma}{\phi}\in\Reals\cup\{+\infty,-\infty\}\) of STL are standard. See \cref{app:stl}.

    \subsection{Finite Variability}

       Finite variability is a standard assumption used to make STL model checking and synthesis tractable in practice (see, \eg~\cite{lee2021efficient,prabhakar2018automatic}). In particular, without such a restriction the full (unbounded) STL model-checking problem is EXPSPACE-complete \cite{alur1996benefits}; adopting finite variability permits efficient bounded analyses, which we also assume in this work.

        \begin{definition}[finite variability \cite{rabinovich1998decidability}]
            A Boolean signal \(q:\NonNegReals\to \Bool\) has \emph{\(N\)-bounded variability} if there exists a collection \(\mathcal{P} = \{J_i\}_{i=1}^N\) of nonempty, pairwise disjoint intervals with \(\bigcup_{i=1}^N J_i=\NonNegReals\) such that \(q\) is constant on every \(J_i\in\mathcal{P}\).
        \end{definition}

        Intuitively, \(N\)-bounded variability means that the Boolean signal can change its value at most \(N-1\) times over \(\NonNegReals\). We say that a signal \(\sigma\) has $N$-bounded variability \emph{with respect to} an STL formula \(\phi\) if the Boolean signal \(t\mapsto (\sigma^t\vDash\phi)\) has the \(N\)-bounded variability. Moreover, \(\sigma\) has \emph{hereditary N-bounded variability} (abbreviated \(N\)-HBV) with respect to \(\phi\) if \(\sigma\) has \(N\)-bounded variability with respect to every subformula \(\psi\in\sub\).

        We recall the bounded STL trace synthesis problem from~\cite{Sato2024}:
        \begin{problem}\label{prob:trace-syn}
            Given an STL formula \(\phi\) (over \(V\)), a system model \(\model\) (from \(V^\prime\) to \(V\)) with time horizon \(T\), and a variability bound \(N\in\Nats\), find a trace \(\sigma\in\mathcal{L}(\model)\) such that 
            \begin{enumerate*}[label=\roman*)]
                \item \(\sigma\) has the \(N\)-HBV with respect to \(\phi\) and
                \item \(\sigma\vDash\phi\) holds,
            \end{enumerate*}
            or prove such \(\sigma\) does not exist.
        \end{problem}

        In this paper we study a \emph{diversity-aware} variant of the above problem: instead of returning an arbitrary satisfying trace, we generate a set of traces that are mutually diverse, \ie sufficiently distinct under a suitable notion of difference.
        


        Formalizing diversity is nontrivial: it can be quantified in multiple, fundamentally different ways. In \cref{sec:diver} we discuss this issue and present three distinct methods for measuring and enforcing diversity in MILP-based synthesis.

\section{MILP-Based Trace Synthesis}\label{sec:milp}

    We recall the MILP-based trace synthesis workflow in~\cite{Sato2024}. The key idea, originally from~\cite{bae2019bounded},  is to \emph{partition}  the time horizon into time intervals, with respect to the truth values of relevant STL formulas. The partition should be such that
    \begin{itemize}
        \item it ensures \emph{interval-wise constancy}: the truth value of each subformula \(\psi\in\sub\) remains constant on each interval; \label{item:constancy}
        \item it is \emph{MILP-encodable}: it avoids strict inequalities (\(<\)), since MILP solvers cannot distinguish \(<\) from \(\le\) and only support non-strict ones (\(\le\))~\cite{gurobi2019gurobi}. \label{item:encodable}
    \end{itemize}
    The original SMT-based approach~\cite{bae2019bounded} partitions the horizon into alternating singleton and open intervals, \ie \(\dotsc,(\gamma_{i-1},\gamma_{i}),\{\gamma_{i}\},(\gamma_{i},\gamma_{i+1}),\dotsc\).
    However, this scheme requires \(<\) to distinguish open and closed intervals and is therefore not MILP-encodable. To address this, we use closed intervals \(\dotsc,[\gamma_{i-1},\gamma_i],[\gamma_i,\gamma_{i+1}],\dotsc\) to partition the horizon.

    Closed intervals can be expressed only using \(\le\); but they have overlaps (at \(\gamma_i\)). This overlap creates ambiguity for the previous notion of interval-wise constancy, which requires that a subformula is either ``constantly true'' or ``never true'' on an interval. The technique of \emph{\(\delta\)-stable partition}~\cite{Sato2024} avoids this ambiguity via introducing the \emph{\(\delta\)-tightening} of a STL formula with some \(\delta>0\) (\cref{def:delta-stren}). 

    \begin{definition}[time sequence, time state sequence~\cite{Sato2024}]\label{def:timeseq}
        A \emph{time sequence} over \([0, T ]\) is a sequence \(\Gamma = (0 = \gamma_0 < \dotsb < \gamma_N = T )\). The same construct is often denoted by \(\Gamma = ([\gamma_{i-1}, \gamma_i])_{i=1}^N\), using the corresponding partition of \([0,T]\) into closed intervals with singular overlaps.  We write \(\Gamma_i\) for \([\gamma_{i-1},\gamma_i]\).

        Given a time sequence, a \emph{time state sequence} over \(V\) is a sequence \(\varsigma=\big((x_0,\gamma_0),\dotsc,(x_N,\gamma_N)\big)\), where \(x_0,\dotsc,x_N\) in \(\Reals^{\card{V}}\).
    \end{definition}
    
    \begin{definition}[piecewise-linear signal~\cite{Sato2024}]\label{def:pwlsig}
        Given a timed state sequence \(\varsigma = \big((x_0, \gamma_0), \dotsc ,(x_N , \gamma_N )\big)\), the piecewise-linear signal \(\pwlsig \colon [0, \gamma_N ] \to \Reals^{\card{V}}\) is defined by the following linear interpolation: \(\pwlsig \ddef (1 - \lambda) x_{i-1} + \lambda x_i\) if \(\gamma_{i - 1} \le t \le \gamma_i ~\big( \text{where } \lambda =  \frac{1}{\gamma_i-\gamma_{i-1}}(t - \gamma_{i - 1})\big)\).
    \end{definition}

    \begin{definition}[\(\delta\)-tightening of STL formulas in NNF~\cite{Sato2024}]\label{def:delta-stren}
        Let \(\phi\) be an STL formula in NNF whose atomic propositions are linear predicates of the form \((c^\top x + b \ge 0)\). Given \(\delta>0\), the \emph{\(\delta\)-tightening} \(\varphi^\delta\) is obtained from \(\varphi\) by 
        \begin{enumerate*}[label=\roman*)]
            \item replacing every occurrence of \((c^\top x + b \ge 0)\) with \((c^\top x + b \ge \delta)\), and
            \item replacing every occurrence of \(\neg(c^\top x + b \ge 0)\) with \((-c^\top x - b \ge \delta)\).
        \end{enumerate*}
        All Boolean connectives and temporal connectives (with the time intervals therein) are left unchanged.
    \end{definition}

    Note that \(\phi^\delta\) is stronger than \(\phi\), \ie~\(\big(\sigma \vDash \phi^\delta\big) \Rightarrow \big(\sigma \vDash \phi\big)\) for all signals \(\sigma\). Using this notion, a \(\delta\)-stable partition defined below classifies each subformula on each closed interval as either  
    \begin{enumerate*}[label=\roman*)]
        \item always true, or
        \item never \emph{robustly} true.
    \end{enumerate*}


    \begin{definition}[\(\delta\)-stability \cite{Sato2024}]
        Let \(\phi\) be an STL formula over \(V\) in NNF, \(\sigma\in\signal_V^T\) be a signal, and \(\Gamma=(\gamma_0,\dotsc,\gamma_N)\) be a time sequence (\cref{def:timeseq}) with \(\gamma_N=T\). We say \(\Gamma\) is \emph{\(\delta\)-stable} for \(\sigma\) and \(\phi\) if, for each \(i=1,\dotsc,N\) and each subformula \(\psi\in\sub\), either of the following holds: 
        \begin{enumerate*}[label=\roman*)] 
            \item \(\sigma^t \vDash \psi\) for each \(t\in \Gamma_i\), or 
            \item \(\sigma^t \nvDash \psi^\delta\) for each \(t\in \Gamma_i\).
        \end{enumerate*}
    \end{definition}    

    Finding a \(\delta\)-stable partition \(\Gamma\) for a signal \(\sigma\) and formula \(\phi\), so that \(\sigma^t\vDash\phi\) holds for \(t\in\Gamma_1\), is a sufficient method to prove \(\sigma\vDash\phi\). \rev{Recall that $\Gamma_1=[\gamma_0,\gamma_1]$ with $\gamma_0=0$ (see \cref{def:timeseq}).}
    As shown in~\cite{Sato2024}, this search can be encoded as an MILP when \(\sigma\) is piecewise-linear (\cref{def:pwlsig}). 
    

    \begin{definition}[conservative valuation~\cite{Sato2024}]\label{def:valuation}
        Let \(\phi\) be an STL formula in NNF, and \(\Gamma=(\gamma_0,\dotsc,\gamma_N)\) be a time sequence. A \emph{valuation} of \(\phi\) in \(\Gamma\) is a function \(\Theta\colon\sub\times\{1,\dotsc,N\}\to\Bool\) that assigns, to each subformula and index of the intervals of \(\Gamma\), a Boolean truth value. Let \(\sigma\) be a signal with a time horizon \(T=\gamma_N\). We say that \(\Theta\) is a \emph{conservative valuation} of \(\phi\) in \(\Gamma\) on \(\sigma\) (up to \(\delta\)) if 
        \begin{enumerate*}[label=\roman*)]
            \item \(\Theta(\psi,i)=\top\) implies that, for each \(t\in\Gamma_i\), \(\sigma^t\vDash \psi\) holds; and
            \item \(\Theta(\psi,i)=\bot\) implies, for each \(t\in\Gamma_i\), \(\sigma^t\nvDash \psi^\delta\).
        \end{enumerate*}
    \end{definition}
    
    \begin{lemma}[\cite{Sato2024}]
        Suppose there is a conservative valuation \(\Theta\) of an STL formula \(\phi\) in a time sequence \(\Gamma\) on  \(\sigma\) up to \(\delta\). Then \(\Gamma\) is \(\delta\)-stable for \(\sigma\) and \(\phi\).
    \end{lemma}

    We simply write \(\valu{\psi}_i\) for \(\Theta(\psi,i)\) when \(\Theta\) is clear from context. The MILP problem therefore searches for
    \begin{enumerate*}[label=\roman*)]
        \item a partition \(\Gamma = ([\gamma_{i-1},\gamma_i])_{i=1}^N\) (see \cref{def:timeseq}), and
        \item a conservative valuation \(\Theta\) (see \cref{def:valuation}),
    \end{enumerate*}
    such that:
    \begin{itemize}
        \item the truth values assigned to subformulas comply with the STL semantics (\cref{subsec:stl}), \eg for a subformula \((\psi_1\wedge\psi_2)\in\sub\) and \(i\in\{1,\dotsc,N\}\), \(\valu{\psi_1\wedge\psi_2}_i = \valu{\psi_1}_i \wedge \valu{\psi_2}_i\);
        \item \(\Theta\) assigns \(\top\) to the top-level formula \(\phi\) in \(\Gamma_1\); and
        \item there is a piecewise-linear trace \(\sigma\in\mathcal{L}(\model)\) of \(\model\) that yields \(\Theta\).
    \end{itemize}
    The entities \(\Gamma, \Theta\) we search for are expressed as MILP variables, and the above three conditions are expressed as MILP constraints. The variables can be classified into four types, where \(N\in\Nats\) is a constant for variability bound (\cref{prob:trace-syn}). 
    \begin{itemize}
        \item Real-valued variables \(\gamma_0,\dotsc,\gamma_N\) specifying the time sequence \(\Gamma\).
        \item Boolean variables \(\{\asg[\psi]_i \mid 1\le i \le N, \psi\in\sub\}\) for the value \(\Theta(\psi,i)\) of a valuation \(\Theta\) that we search for.
        \item Real-valued variables \(\{x_{i,v}\mid 0\le i \le N, v\in V\}\) for the values of a piecewise-linear trace \(\sigma\in\mathcal{L}(\model)\).
        \item Auxiliary variables (real or binary) used to encode big-\(M\) constraints---a common technique in LP---and the temporal/satisfaction encodings.
    \end{itemize}

    We omit the MILP details here; see \cite{Sato2024} for the full encoding.

\section{Diversity-Aware Synthesis}\label{sec:diver}

    This section presents methods to synthesize mutually diverse traces. Although modern MILP solvers such as Gurobi~\cite{gurobi} can return multiple solutions from the solution pool of a single instance, those solutions are typically very similar (see \cref{sec:exp} for an empirical demonstration).
        
    To overcome this limitation, we propose three diversity measures that quantify how ``different'' one trace is from another. Each measure is encoded as an MILP objective to promote diverse solutions. \rev{We do not claim that the three methods presented here are the only possible approaches. However, given the practical need to generate diverse example traces, it is important to provide methods that work.}
        
    \subsection{Boolean Distance Method (BD)}\label{subsec:bd}
        \rev{As flipping the truth value of a subformula may induce substantial changes in the underlying signal variables, an initial attempt to diversify traces is to bias the truth valuations of subformulas in each trace. We quantify the resulting diversity using a Boolean distance between traces.}



        \begin{definition}[Boolean distance]
            Let \(\sigma_1, \sigma_2\) be two traces found by an MILP solver that 
            \begin{enumerate*}[label=\roman*)]
                \item have \(N\)-HBV (see \cref{prob:trace-syn}) and 
                \item satisfy the same STL formula \(\phi\).
            \end{enumerate*} 
            Let \(\Gamma^{\sigma_1}\) (\resp~\(\Gamma^{\sigma_2}\)) be the \(\delta\)-stable partition for \(\sigma_1\) (\resp~\(\sigma_2\)) and \(\phi\). Let \(\valu{\psi}_i^{\sigma_j}\) denote the Boolean valuation of \(\psi\in\sub\) over the interval \(\Gamma_i^{\sigma_j}\) of the partition \(\Gamma^{\sigma_j}\), where \(i\in\{1,\dotsc,N\}\) and \(j\in\{1,2\}\). The \emph{Boolean distance} between \(\sigma_1\) and \(\sigma_2\) is
            \begin{align}\label{eq:boolean-dist}
                \bdist(\sigma_1,\sigma_2) \ddef \sum_{i=1}^{N}\sum_{\psi\in\sub} \abs{\valu{\psi}_i^{\sigma_1}-\valu{\psi}_i^{\sigma_2}}.
            \end{align}
        \end{definition}
        By construction, \(\bdist(\sigma_1,\sigma_2)\) is equal to the number of pairs \((i,\psi)\) on which the two traces assign different truth values. Hence, a larger Boolean distance indicates greater diversity, in the sense of satisfaction of subformulas.

        
        \begin{remark}\label{rmk:affine}
            In \cref{eq:boolean-dist}, if the Boolean value of \(\valu{\psi}_i^{\sigma_1}\) is fixed for \(i\in \{1,\dotsc,N\}\) (which is the case in our usage), the absolute value  can be rewritten as 
            \begin{align*}
              &  \abs{\valu{\psi}_i^{\sigma_1}-\valu{\psi}_i^{\sigma_2}} =
                \begin{cases}
                    \valu{\psi}_i^{\sigma_2} & \tif \valu{\psi}_i^{\sigma_1}=0,\\[2pt]
                    1-\valu{\psi}_i^{\sigma_2} & \tif \valu{\psi}_i^{\sigma_1}=1,
                \end{cases}
            \quad
            \\
               & \text{
            or equivalently,} \quad
               \abs{\valu{\psi}_i^{\sigma_1}-\valu{\psi}_i^{\sigma_2}} \,=\, (1-2 \valu{\psi}_i^{\sigma_1})\, \valu{\psi}_i^{\sigma_2}+ \valu{\psi}_i^{\sigma_1}.
            \end{align*}
            Note that the last expression is affine in \(\valu{\psi}_i^{\sigma_2}\) 
            (\(\valu{\psi}_i^{\sigma_1}\) is a constant).
            Applying this rule to~\cref{eq:boolean-dist}, we obtain an affine expression over variables \(\valu{\psi}_i^{\sigma_2}\).
        \end{remark}

        
        Given a fixed variability bound \(N\in\Nats\), we incorporate the Boolean distance of traces into the MILP objective while keeping the original constraints from \cref{sec:milp} unchanged. This MILP problem is then solved iteratively, for $m=1,2,\dotsc$, to synthesize a trace $\sigma_m$ that is maximally distinct (in terms of Boolean distance) from the previously synthesized traces $\sigma_1,\dotsc,\sigma_{m-1}$. The objective in each iteration is defined as follows.
        
        \begin{itemize}
            \item \textbf{Initialization ($m=1$).} Maximize $f_1\ddef 0$.
            \item \textbf{Iteration ($m>1$).} Let \(\sigma_m\) denote the trace to be synthesized at the current iteration (this gives decision variables). Given the previously obtained traces \(\{\sigma_i\}_{i=1}^{m-1}\) (which are constants in MILP), maximize 
                \begin{align}\label{eq:m-obj}
                    f_m \ddef \sum_{i=1}^{m-1}  \bdist\left(\sigma_m,\sigma_i\right).
                \end{align}
        \end{itemize} 
        
        By \cref{rmk:affine}, $f_m$ is an affine function of the Boolean decision variables of the trace \(\sigma_m\), with coefficients determined by the constants for \(\{\sigma_i\}_{i=1}^{m-1}\).


        \begin{remark}\label{rmk:MILP-good}
            The above diversification method is possible due to our use of MILP. A similar method in SMT is nontrivial. The key difference is that, in MILP, truth values are encoded as integer variables and they can be summed up (as in \cref{eq:m-obj}), whereas in SMT they are propositional variables and cannot be summed. Counting the number of subformulas with different truth values across models would require techniques such as MaxSMT or model counting, which we leave for future work.
        \end{remark}

    \subsection{Randomized Boolean Distance Method (RBD)}\label{subsec:rbd}

        This subsection introduces a randomized variant of the Boolean distance. It promotes diversity by evaluating each candidate trace against an independently sampled reference, thereby decoupling the current solution from all previous ones.
        

        Specifically, for a fixed variability bound \(N\), the valuation of all \(\psi\in\sub\) across the intervals \((\Gamma_i)_{i=1}^N\) form an \(N\card{\sub}\)-dimensional Boolean vector. We sample a \emph{reference valuation} \(\reftrace\) from the uniform distribution in \(\Bool^{N\card{\sub}}\) and solve an MILP problem that maximizes the Boolean distance between \(\reftrace\) and the candidate trace. Let \(\sigma\) denote the candidate trace, the objective is
        \[
            f \ddef \bdist(\sigma,\reftrace),
        \]
        which yields a trace that is as far as possible, in Boolean distance, from the reference valuation. By \cref{rmk:affine}, \(f\) is affine in the decision variables of trace \(\sigma\).

        \begin{remark}
            The choice between maximization and minimization is immaterial: for any reference valuation \(\reftrace \in \Bool^{N\,\card{\sub}}\), 
            \[ 
                \bdist\,(\sigma,\reftrace)=N\,|\sub|-\bdist(\sigma,\bar{\sigma}_{\mathrm{ref}}),
            \]
            where \(\bar{\sigma}_{\mathrm{ref}}\) denotes the bitwise complement of \(\reftrace \). Hence, maximizing distance from \(\reftrace\) is equivalent to minimizing distance to \(\bar{\sigma}_{\mathrm{ref}}\).
        \end{remark} 

        The following proposition provides a quantitative guarantee that the probability of generating identical reference traces is small, which supports the diversity of traces. The proof is deferred to \cref{appendix:proof}.

        \begin{restatable}{proposition}{propDiv}
            For a \(\delta\)-stable partition \(\Gamma=(\gamma_0, \cdots, \gamma_N)\) of the time horizon for a formula \(\phi\), let \(d \ddef N\,|\sub|\). Suppose \(m\ge 2\) reference valuations \(\reftrace^1,\dotsc,\reftrace^m\) are drawn independently and uniformly from $\Bool^{d}$. Let $p$ denote the probability that at least two references coincide. Then
            \[
                p \;=\; 1 - \frac{2^d!}{(2^d-m)! \, 2^{md}}\raisebox{-2ex}{,}
            \]
            and the following upper bound holds:
            \[
                p\:\le \:\frac{m(m-1)}{2^{d+1}}\raisebox{-2ex}{.}
            \]
        \end{restatable} 

        \subsection{Value Distance Method (VD)}\label{subsec:vd}

        \rev{Although the Boolean distance can serve as a metric for diversification, we also consider a more direct metric:} the distance between the trace values. For two piecewise-linear traces that share the same time sequence, their \emph{value distance} is defined as follows.

        \begin{definition}
            Let \(\varsigma_1,\varsigma_2\) be two time state sequences (\cref{def:timeseq}) with the same time sequence \(\Gamma = (\gamma_0, \dotsc, \gamma_N)\). Let \(V\) be a set of signal variables. For two piecewise-linear traces \(\sigma_1, \sigma_2: [\gamma_0,\gamma_N] \to \Reals^{\card{V}}\) generated from \(\varsigma_1,\varsigma_2\), respectively, their \emph{value distance} is defined by
            \[
                \vd(\sigma_1,\sigma_2) \ddef \sum_{i=0}^{N}\lVert\sigma_1(\gamma_i) - \sigma_2(\gamma_i)\rVert_1 = \sum_{i=0}^{N}\sum_{v\in V}\abs{\sigma_1(\gamma_i)(v) - \sigma_2(\gamma_i)(v)} \raisebox{-.8ex}{.}
            \]
            Here $\lVert\cdot\rVert_1$ denotes the $\ell_1$-norm on $\Reals^{\card{V}}$.
        \end{definition}

        Essentially, the value distance is the cumulative distance between traces at \((\gamma_i)_{i=0}^N\) of the time sequence \(\Gamma\). 
        
        \begin{remark}
            \rev{Note that, although the time partitions $\Gamma$ may vary across traces for BD and RBD (\ie $\Gamma^{\sigma_i}$ may differ from $\Gamma^{\sigma_j}$), they are fixed for VD. This is because allowing variable partitions in VD may reduce meaningful diversity: two identical signals could appear artificially distant if compared at different time instants.}
        \end{remark}
        The diversified trace synthesis procedure is iterative: at each step it maximizes the value distance between the candidate trace and all previously generated traces, subject to the constraint that all traces share the same time sequences. Let \(\sigma_m\) denote the piecewise-linear signal produced at iteration $m$, and let $(\gamma_0^{(m)},\ldots,\gamma_N^{(m)})$ be the corresponding $\delta$-stable partition at iteration \(m\). The iterative synthesis procedure is as follows.

        \begin{itemize}
            \item \textbf{Initialization ($m=1$).} Solve the original MILP problem with objective $f_1\ddef 0$.
            \item \textbf{Iteration ($m>1$).} Let $\mathcal{C}$ denote the original MILP constraints from \cref{sec:milp}. Solve the following MILP:
            \begin{align}\label{eq:value-dif}
                \text{maximize } &\sum_{k=1}^{m-1}\sum_{i=0}^N\sum_{v\in V} Z_{k,i,v}, \\
                \text{subject to } & (\C), \nonumber \\
                & \gamma_i = \gamma_i^{(m-1)} & 0 \leq i \leq N, \nonumber\\
                & Z_{k,i,v} = \texttt{abs}(Y_{k,i,v}) & 1\leq k \leq m-1 , 0\leq i\leq N, v\in V, \nonumber\\
                & Y_{k,i,v} = x_{i,v} - {\sigma}_k(\gamma_i)(v) & 1\leq k \leq m-1 , 0\leq i\leq N, v\in V.\nonumber
            \end{align}
            Here, \(x_{i,v}\) is the MILP decision variable representing the value of signal variable \(v\in V\) on the interval \([\gamma_{i-1},\gamma_i]\).
        \end{itemize}

        At iteration \(m>1\), to 
        \begin{enumerate*}[label=\roman*)]
            \item ensure that synthesized traces share the same \(\delta\)-stable partition, and 
            \item simplify the encoding of absolute values in the objective, 
        \end{enumerate*}
        we augment \(\mathcal{C}\) with the constraints above. The relations \(Z_{k,i,v}=\mathtt{abs}(Y_{k,i,v})\) are implemented by standard MILP linearization (introducing auxiliary variables and linear inequalities), so the resulting problem remains an MILP.

        \begin{remark}
            Due to the use of absolute-value terms in the constraints, the MILP complexity increases with each iteration. Specifically, at iteration \(m\), new variables \(Y_{m,i,v}\) and \(Z_{m,i,v}\) are introduced for every \(i=0,\dots,N\) and \(v\in V\), adding \(2(N+1)|V|\) variables per iteration. Consequently, the overall complexity of generating \(m\) solutions grows faster than linearly with respect to the size of the individual MILP problem.
        \end{remark}

        Nevertheless, the value distance method is substantially faster than the two Boolean distance approaches when synthesizing a small number of traces. The value distance objective sums over signal variables, whereas the Boolean distance objective sums over subformulas. Since \(\card{V}\) is typically much smaller than \(\card{\sub}\) for realistic STL formulas, the value distance objective contains far fewer terms and is therefore cheaper to optimize. See \cref{sec:exp} for experimental results.

    \begin{figure}[t]
  \centering
  \begin{subfigure}[t]{1\textwidth}
    \centering
    \begin{lstlisting}
 import gurobipy as gp
 from stltspref.problem import ~create_stl_milp_problem~

 bound = 10                           # variability bound (see Problem 1)
 milp = gp.Model()                            # Gurobi model
 prob = ~create_stl_milp_problem~(milp,N=bound) # initialize the MILP problem 
    \end{lstlisting}
    \vspace{-2mm}
    \caption{Create the MILP problem \texttt{prob}.}
\end{subfigure}
\vspace{2mm}
\begin{subfigure}[t]{1\textwidth}
    \begin{lstlisting}
 # Define system model and its continuous states
 car = prob.~create_system_model~()
 car.~add_state~('x', 0, 40)            # position x in [0,40]
 car.~add_state~('v', -5, 5)            # velocity v in [-5,5]
 car.~add_state~('a', -3, 3)            # acceleration in [-3,3]

 # Set initial state and dynamics (double-integrator model)
 car.~set_initial_state~(x=(0,0), v=(0,0), a=(0,0))   # initial value ranges
 car.~add_dynamics~('a', -3, 3, constant=True) 
 car.~add_double_integrator_dynamics~('x', 'v', 'a')
    \end{lstlisting}
    \vspace{-2mm}
    \caption{Define the system model \(\dot x = v,\ \dot v = a,\ a\in[-3,3]\), and add it to \texttt{prob}.}
    \label{subfig:model}
  \end{subfigure}
  \vspace{2mm}
  \begin{subfigure}[t]{1\textwidth}
    \centering
    \begin{lstlisting}
 from stltspref.linear_expression import LinearExpression as L
 from stltspref.stl import ~Atomic~, ~BoundedAlw~, ~Ev~

 danger = ~Atomic~(L.@unit@('x') <= 10)   # atomic proposition
 stl_spec = ~Ev~(~BoundedAlw~([0, 5], danger))           # STL specification 

 prob.~initialize_milp_formulation~(stl_spec)          # add the spec to MILP
    \end{lstlisting}
    \vspace{-2mm}
    \caption{Specify the STL formula \(\ev\left(\alw_{[0,5]}x \le 10\right)\) and add it to \texttt{prob}.}
    \label{subfig:spec}
  \end{subfigure}
  \vspace{2mm}
  \begin{subfigure}[t]{1\textwidth}
    \centering
    \begin{lstlisting}
 from stltspref.preferential_synthesis import ~diversity_finder~

 diversity_method = 'RBD'    # random Boolean distance method
 trace_num = 5             # synthesis 5 diverse traces

 ~diversity_finder~(prob, diversity_method, bound, trace_num)
    \end{lstlisting}
    \vspace{-2mm}
    \caption{Solve \texttt{prob} to synthesize diverse traces.}
    \label{subfig:solve}
  \end{subfigure}
  \caption{Workflow for MILP-based trace synthesis using STLts-Div.}
  \label{fig:syn-code}
\end{figure}

\section{Demonstration of STLts-Div}\label{sec:demo}
    
    STLts-Div is a Python package for MILP-based diversified trace synthesis from STL. It can be installed by 
    \begin{enumerate*}[label=\roman*)] 
        \item ensuring a working Gurobi installation, and 
        \item running \texttt{pip install stltspref}.
    \end{enumerate*}
    \rev{We use the name \texttt{stltspref} for the pip package to accommodate future extensions incorporating preference-based trace synthesis methods.}

    The procedure to run STLts-Div is shown in \cref{fig:syn-code}. The user must provide the following input.
    \begin{itemize}
        \item A system model. STLts-Div supports models such as \emph{double-integrator dynamics} and \emph{rectangular hybrid automata} (RHA) (cf.\ \cite{Sato2024} for details); see \cref{subfig:model}.
        \item An STL formula (encoded in a DSL style); see \cref{subfig:spec};
        \item A diversification method: SP (Gurobi solution pool \cite{gurobi}), BD (Boolean distance; \cref{subsec:bd}), RBD (randomized Boolean distance; \cref{subsec:rbd}), or VD (value difference; \cref{subsec:vd}); see \cref{subfig:solve};
        \item The desired number of traces to synthesize (a positive integer); see \cref{subfig:solve}.
    \end{itemize}

\section{Experimental Analysis}\label{sec:exp}
    In this section, we evaluate our tool across different benchmarks and STL specifications to assess the efficiency and effectiveness of each diversification method. As a baseline, we also include the standard approach that extracts traces directly from the solution pool of Gurobi.

    \subsection{Experiment Setting}
        \myparagraph{Benchmarks}
        To evaluate the efficiency of our diversification tool, we used benchmarks from four categories: \emph{Danger-then-Stop} (DSTOP), \emph{Rear-End Near Collision} (RNC), \emph{Navigation} (NAV), and \emph{Disturbance Scenarios in ISO~34502} (ISO). Each category shares the same system model but includes several distinct STL specifications, resulting in a total of 13 benchmarks.

        \paragraph{Danger then stop (DSTOP)} Consider two cars in a single lane, where the rear car \(\mathrm{Car_r}\) repeatedly approaches the front car \(\mathrm{Car_f}\) too closely, creating a potentially dangerous situation. Eventually, one of the cars comes to a stop. The system dynamics are given by \(\dot x_{\rear}=v_{\rear}\), \(\dot v_{\rear} = a_{\rear}\), \(\dot x_{\front} = v_{\front}\), \(\dot v_{\front} = a_{\front}\), with \(a_{\rear},a_{\front} \in [-3,3]\). The STL specification \(\mathtt{DSTOP}\) is defined as follows:

        \begin{align*}
            \begin{array}{rll}
            & \mathtt{lead} \lddef x_{\front} - x_{\rear} \geq 0,\qquad
            \mathtt{close} \lddef x_{\front} - x_{\rear} \leq 10 , \qquad
            \mathtt{far} \lddef x_{\front} - x_{\rear} \geq 40 , \\
            & \mathtt{DSTOP} \lddef (\alw ~\mathtt{lead}) \wedge (\alw_{[0,2]}\mathtt{far})\wedge \left[(\ev_{[0,10]} \mathtt{close}) \unt_{[10,15]} (v_{\front} \leq 0 \vee v_{\rear} \leq 0)\right]. 
            \end{array}
        \end{align*}

        \paragraph{Rear-End Near Collision (RNC1–3)} The system model for RNC1--3 is identical to that of DSTOP. These benchmarks describe driving scenarios in which the rear car \(\mathrm{Car_r}\) approaches the front car \(\mathrm{Car_f}\) too closely. The corresponding STL specifications, \(\mathtt{RNC1}\), \(\mathtt{RNC2}\), and \(\mathtt{RNC3}\), are defined as follows:
        \begin{align*}
        \begin{array}{rl}
            \mathtt{danger}     & \quad\lddef\quad x_{\front} - x_{\rear} \leq 10 ,                                                                                               \\
            \mathtt{dyn\_{inv}} & \quad\lddef\quad  x_{\front} - x_{\rear} \ge 0\,\land\, 2 \leq v_{\front} \leq 27 \,\land\, 2 \leq v_{\rear} \leq 27,
            \\
            \mathtt{trimming}   & \quad\lddef\quad  (\ev \mathtt{danger}) \Rightarrow \bigl((\alw_{[0, 0.2]} a_{\rear} \geq 0.5) \unt \mathtt{danger}\bigr),
            \\
            \mathtt{trimming2} & \quad\lddef\quad (\Diamond \mathtt{danger}) \Rightarrow \bigl((\Box_{[0, 1]} a_{\rear} \geq 1) \mathbin{\mathcal{U}} \mathtt{danger}\bigr),
        \\
            \mathtt{RNC1}       & \quad\lddef\quad  \alw(\mathtt{dyn\_inv}\land \mathtt{trimming}) \land \ev_{[0, 9]}\alw_{[0,1]}\mathtt{danger},\\
            \mathtt{RNC2}      & \quad\lddef\quad  \bigl(\Box (x_{\front} - x_{\rear} \ge 0)\bigr) \land
        \\
                           & \qquad \qquad
        \Diamond_{[0, 9]} \bigl(
        (\Box_{[0,1]} \mathtt{danger} ) \land
        (\Box_{[0,1]} a_{\rear} \ge 1 ) \land
        (\Diamond_{[1,5]} \lnot\mathtt{danger})
        \bigr)  ,   \\
        \mathtt{RNC3}      & \quad\lddef\quad  \Box(\mathtt{dyn\_inv}\land\mathtt{trimming2} ) \land \Diamond_{[0, 9]}\Box_{[0,1]}\mathtt{danger}.
            \end{array}
        \end{align*} 

        {\makeatletter
        \let\par\@@par
        \par\parshape0
        \everypar{}
    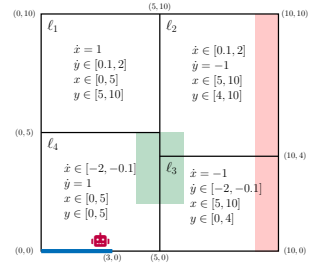
\begin{wrapfigure}[10]{r}[0pt]{4cm}
    \centering
    \vspace{-2.5em}
    \centering
    \begin{adjustbox}{max width = 1\linewidth}
        \begin{tikzpicture}
            \fill[fill=red!30, fill opacity=0.7] (9,10) rectangle (10,0); 
            \fill[fill=ForestGreen!40, fill opacity=0.7] (4,5) rectangle (6,2); 

            \draw[line width=1.5pt] (0,10) rectangle (10,0);
            \draw[line width=1.5 pt] (0,5) -- (5,5);
            \draw[line width=1.5 pt] (5,4) -- (10,4);
            \draw[line width=1.5 pt] (5,0) -- (5,10);
            \node at (-0.7, 10) {$(0,10)$};
            \node at (-0.7, 0) {$(0,0)$};
            \node at (10.7, 0) {$(10,0)$};
            \node at (10.7, 10) {$(10,10)$};
            \node at (-0.7,5) {$(0,5)$};
            \node at (10.7,4) {$(10,4)$};
            \node at (5,-0.3) {$(5,0)$};
            \node at (5,10.3) {$(5,10)$};
            \node at (3,-0.3) {$(3,0)$};

            \draw[line width=5pt, color = NavyBlue] (0,0) -- (3,0); %
            \node[color=purple] at (2.5 ,0.5) {\LARGE \faRobot};

            \node at (0.5, 9.5) {\LARGE $\ell_1$};
            \node[align=left, font=\Large] at (2.5, 7.5) {$\dot{x}=1$ \\ $\dot{y} \in [0.1,2]$ \\ $x\in[0,5]$ \\ $y \in [5,10]$};

            \node at (5.5, 9.5) {\LARGE $\ell_2$};
            \node[align=left, font=\Large] at (7.5, 7.5) {$\dot{x} \in [0.1, 2]$ \\ $\dot{y} = -1$ \\ $x\in[5,10]$ \\ $y \in [4,10]$};

            \node at (5.5, 3.5) {\LARGE $\ell_3$};
            \node[align=left, font=\Large] at (7.8, 2.3) {$\dot{x} = -1$ \\ $\dot{y} \in [-2, -0.1]$ \\ $x\in[5,10]$ \\ $y \in [0,4]$};

            \node at (0.5, 4.5) {\LARGE $\ell_4$};
            \node[align=left, font=\Large] at (2.5, 2.5) {$\dot{x} \in [-2,-0.1]$ \\ $\dot{y}=1$ \\ $x\in[0,5]$ \\ $y \in [0,5]$};
        \end{tikzpicture}
    \end{adjustbox}

    \vspace{-0.6em}
    \caption{The NAV1--2 RHA}
    \label{fig:example_nav}
    \end{wrapfigure}

    \paragraph{Navigation (NAV1–2)} The system model for NAV1–2 is adapted from~\cite{duggirala2011abstraction}, a standard example of a rectangular hybrid automaton (RHA) widely used in prior work (\eg~\cite{Sato2024,bu2022arch}). The model describes the motion of a robot on a grid consisting of four rectangular regions \(\ell_1,\dotsc,\ell_4\). Each region has its own dynamics (see \cref{fig:example_nav}). The environment includes an unsafe region \(\mathtt{unsafeR}\) (red) and a goal region \(\mathtt{goalR}\) (green). The robot starts from an initial position \((x_0,y_0)\) where \(x_0\in[0,3]\) and \(y_0=0\) (blue line). The STL specifications are $\mathtt{NAV1}\,:\equiv\,\Diamond (\Box_{[0,3]} ((x, y) \in \mathtt{goalR})) \land \Box( x \not\in \mathtt{unsafeR})$

    \par}%

    \noindent
    and $\mathtt{NAV2}\,\equiv\,\Box((x,y) \in \ell_3 \Rightarrow \Diamond_{[0,3]} (x,y) \in \ell_4)$. Here, \(\mathtt{NAV1}\) specifies a reach-avoid property, while \(\mathtt{NAV2}\) requires that whenever the robot enters region \(\ell_3\), it must reach region \(\ell_4\) within three seconds.
        
    \paragraph{Disturbance Scenarios in ISO 34502 (ISO1, ISO3, $\dotsc$, ISO8)} ISO 34502 is an international standard for testing autonomous vehicles. The system model is similar to that of RNC1–3, and the specifications \(\mathtt{ISO1}, \mathtt{ISO3}, \dotsc, \mathtt{ISO8}\) are obtained in~\cite{reimann2024temporal}. ISO2 is omitted because it involves three vehicles. Each spec \(\mathtt{ISO}i\) follows the general structure defined in~\cite{reimann2024temporal}:
    \begin{displaymath}\small
    \begin{array}{rcl}
        \mathtt{ISO}{i}
         & \;\lddef\; & \mathtt{initSafe}
        \wedge \mathtt{disturb}_i,
        \qquad                            \\
        \mathtt{disturb}_i
         & \;\lddef\; &
        \mathtt{initialCondition}_i
        \wedge \mathtt{behaviourSV}_i
        \wedge \mathtt{behaviourPOV}_i,
    \end{array}
\end{displaymath}
where SV denotes the subject (``ego'') vehicle and POV the principal other vehicle. The subformulas \(\mathtt{initialCondition}_i\), \(\mathtt{behaviourSV}_i\), and \(\mathtt{behaviourPOV}_i\) vary across scenarios (indexed by \(i\)). Their detailed definitions are beyond the scope of this paper; please refer to~\cite{reimann2024temporal} for full descriptions.

        \myparagraph{Experiment design}
            We evaluate four diversification strategies implemented in our tool—SP (solution pool), BD (Boolean distance), RBD (randomized Boolean distance), and VD (value distance)—for synthesizing diverse traces. For each benchmark and each method, we generate \(8\) traces. We set the variability bound as \(10\) uniformly for each benchmark. \rev{While measuring diversity via coverage metrics over the space of all possible traces would provide a comprehensive quantification of diversity, the infinite nature of the trace space (since signals evolve in continuous time) makes such an approach computationally intractable. Investigating tractable approximations is left for future work.} Instead, we quantify diversity using the aggregate pairwise Euclidean distance (APED) over time, a common choice in the literature~\cite{cheng2023behavexplor}:
            \[
                \Deuc \;\ddef\; \sum_{i<j}\int_{0}^{T} \bigl\lVert \sigma_i(t) - \sigma_j(t) \bigr\rVert_{2}\,\dif t,
            \]
            where \(i\) and \(j\) index the traces (\(i,j\in\{1,\dots,8\}\) in our experiments) and \(T\) is the time horizon. \rev{The synthesized traces are resampled over $[0,T]$ with a time step of $0.02$,} and the integral is computed numerically in Python.
            
            
            Our experiments were conducted on an Amazon EC2 c5.2xlarge instance (3.00\,GHz Intel Xeon Platinum 8275CL, 16\,GB RAM) running Ubuntu Server 20.04\,LTS.

            \subsection{Evaluation}
       
        We evaluate our tool along from \emph{effectiveness} and \emph{efficiency}.

        \myparagraph{Effectiveness}
        \Cref{tab:to360} reports APED for each method and benchmark under a 6-minute per-trace timeout; a dash ``---'' indicates that no feasible trace was found within the timeout. Across benchmarks, our diversification methods (BD, RBD, VD) consistently achieve substantially higher diversity than the baseline (SP)—typically about an order of magnitude larger APED. 

\begin{table*}[t!]
  \centering
  \caption{Aggregated pairwise Euclidean distance under two timeout settings.}
  \label{tab:two-subtables}
  \resizebox{\textwidth}{!}{%
  \begin{subtable}[t]{1\textwidth}
    \centering
    \caption{Timeout per trace: \SI{360}{s}.}
    \label{tab:to360}
    \begin{tabular}{
      >{\raggedright\arraybackslash}p{2cm} 
      >{\centering\arraybackslash}p{2.5cm}    
      >{\centering\arraybackslash}p{2.5cm}
      >{\centering\arraybackslash}p{2.5cm}
      >{\centering\arraybackslash}p{2.5cm}
    }
      \toprule
      \textbf{Benchmarks} & \textbf{SP} & \textbf{BD} & \textbf{RBD} & \textbf{VD} \\
      \midrule
      \dstop & \aboutk{129}\ssec{4.62} & \aboutm{1.6}\ssec{200} & \aboutk{934}\ssec{1178} & \aboutm{3}\ssec{107}\\
      \rnc   & 0\ssec{3.11} & \aboutm{1.3}\ssec{2207} & \aboutm{1.5}\ssec{1881} & \aboutm{1.6}\ssec{44.0} \\
      \rncc & 0\ssec{5.87} & \aboutm{2.8}\ssec{2523} & \aboutm{2.2}\ssec{1867} & \aboutm{4.8}\ssec{25.4} \\
      \rnccc & 0\ssec{3.13} & \aboutm{1.7}\ssec{1473} & \aboutm{1.3}\ssec{158} & \aboutm{1.5}\ssec{53.0} \\
      \nav & --- & --- & --- & --- \\
      \navv & --- & --- & --- & --- \\
      \iso & \about{18}\ssec{95.9} & \aboutk{466}\ssec{2529} & \aboutk{403}\ssec{2882} & \aboutm{1.5}\ssec{2529} \\
      \isot & \about{153}\ssec{63.8} & \aboutk{701}\ssec{2527} & \aboutk{71}\ssec{2882} & \aboutk{792}\ssec{2528} \\
      \isofo & \about{772}\ssec{29.6} & \aboutk{633}\ssec{2526} & \aboutk{174}\ssec{2881} & \aboutk{598}\ssec{2526} \\
      \isofi & \aboutk{34}\ssec{266} & \aboutm{1}\ssec{2561} & \aboutk{32}\ssec{2882} & \aboutm{1.7}\ssec{2361} \\
      \isosi & \about{122}\ssec{19.1} & \aboutk{413}\ssec{2529} & --- & \aboutm{1.6}\ssec{1653} \\
      \isose & --- & \aboutk{230}\ssec{2556} & --- & \aboutk{911}\ssec{2555} \\
      \isoe & \aboutk{1.2}\ssec{66.8} & \aboutk{447}\ssec{2527} & \aboutk{291}\ssec{2882} & \aboutk{896}\ssec{2444} \\
      \bottomrule
    \end{tabular}
  \end{subtable}}

    \resizebox{\textwidth}{!}{%
  \begin{subtable}[t]{1\textwidth}
    \centering
    \caption{Timeout per trace: \SI{60}{s}.}
    \label{tab:to60}
    \begin{tabular}{
      >{\raggedright\arraybackslash}p{2cm} 
      >{\centering\arraybackslash}p{2.5cm}    
      >{\centering\arraybackslash}p{2.5cm}
      >{\centering\arraybackslash}p{2.5cm}
      >{\centering\arraybackslash}p{2.5cm}
    }
      \toprule
      \textbf{Benchmarks} & \textbf{SP} & \textbf{BD} & \textbf{RBD} & \textbf{VD} \\
      \midrule
      \dstop & \aboutk{130}\ssec{4.33} & \aboutm{1.6}\ssec{157}& \aboutm{1}\ssec{237} & \aboutm{3}\ssec{102}\\
      \rnc   & 0\ssec{3.03} & \aboutm{1.3}\ssec{405} & \aboutm{1.4}\ssec{401} & \aboutm{1.6}\ssec{42} \\
      \rncc  & 0\ssec{5.52} & \aboutm{2.4}\ssec{423} & \aboutm{2.3}\ssec{477} & \aboutm{4.8}\ssec{24.4} \\
      \rnccc & 0\ssec{3.05} & \aboutm{1.7}\ssec{365} & \aboutm{1.2}\ssec{395} & \aboutm{1.6}\ssec{50.6} \\
      \nav   &  --- & --- & --- & --- \\
      \navv & --- & --- & --- & --- \\
      \iso & --- & \aboutk{79}\ssec{429} & --- & \aboutm{1.4}\ssec{429} \\
      \isot & --- & 0\ssec{428} & --- & \aboutk{133}\ssec{426} \\
      \isofo & \about{772}\ssec{32.1} & \aboutk{177}\ssec{426} & --- & \aboutk{256}\ssec{425} \\
      \isofi & --- & 0\ssec{464} & --- & \aboutm{1.5}\ssec{464} \\
      \isosi & \about{122}\ssec{21.2} & 0\ssec{429} & --- & \aboutm{1.6}\ssec{389} \\
      \isose & --- & 0\ssec{457} & --- & \aboutk{271}\ssec{458} \\
      \isoe & --- & \aboutk{88}\ssec{427} & 0\ssec{480} & \aboutk{431}\ssec{426} \\
      \bottomrule
    \end{tabular}
  \end{subtable}}
  \vspace{0.3em}

  {\footnotesize
  Entries show APED; parentheses give the total time to produce 8 traces (in seconds);
  `---' indicates that no feasible traces were produced within the timeout;
  numbers use suffixes \(k = 10^3\) and \(M = 10^6\);
  }
\end{table*}

Occasionally the baseline attains $\mathrm{APED}>0$, reflecting that some specifications are relatively permissive for the given system model: multiple distinct MILP assignments can satisfy \(\phi\) while respecting the dynamics, yielding non-identical traces. The \dstop, \isofi, and \isoe  exhibit this behaviour.

                    \myparagraph{Efficiency}
        \Cref{tab:to360} reports the total time required to generate the full set of traces. For \nav and \navv no feasible trace was found within our setting; this matches \cite{Sato2024}, which reports that \nav  and \navv becomes solvable only with a much larger variability bound (e.g. \(N=17\) for \nav, versus our uniform bound N=10).

        Because our diversification is posed as an optimization problem, a solver timeout \emph{does not} necessarily imply infeasibility---it may simply return a suboptimal or low-quality solution. To assess sensitivity to the time budget, we repeated the experiments with a reduced per-trace timeout of 60 seconds; results appear in \cref{tab:to60}.

        \begin{figure}[t]
            \centering
            \begin{subfigure}[b]{0.4\textwidth}
                \includegraphics[width=\textwidth]{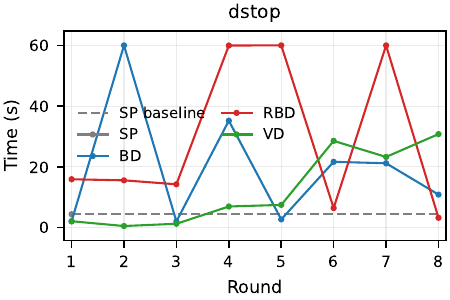}
            \end{subfigure}\hfill
            \begin{subfigure}[b]{0.4\textwidth}
                \includegraphics[width=\textwidth]{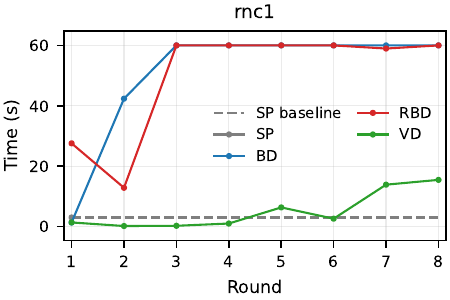}
            \end{subfigure}
            \medskip
            \begin{subfigure}[b]{0.4\textwidth}
                \includegraphics[width=\textwidth]{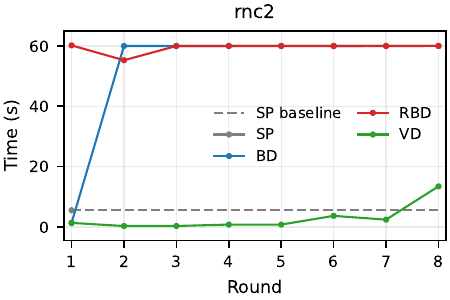}
            \end{subfigure}\hfill
            \begin{subfigure}[b]{0.4\textwidth}
                \includegraphics[width=\textwidth]{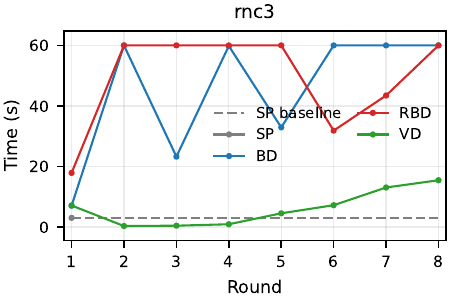}
            \end{subfigure}
            \caption{Round time of \dstop, \rnc, \rncc, \rnccc.}
            \label{fig:round_time}
        \end{figure}

        For relatively simple benchmarks (\eg \dstop, \rnc, \rncc), reducing the timeout has little effect on diversity for any method. \rev{In particular, RBD requires significantly less time while producing nearly the same APED. This suggests that, with a larger time budget (\ie 360s per trace), the optimization process may spend most of its time attempting to approach the optimal solution without reaching it before the timeout. In practice, however, a near-optimal solution already provides sufficient diversity, which explains the similar APED values observed under both timeout settings.} 
        
        By contrast, on more challenging benchmarks (\eg \iso, \isoe, \isofo), performance diverges: RBD often fails to produce feasible traces, and BD frequently returns only trivial repeats after the first iteration ($\mathrm{APED} = 0$). VD is the most robust under tight time limits: it continues to produce nontrivial, diverse traces even with the shorter timeout.

        To study how per-iteration runtime (i.e., the time to synthesize one trace) evolves with the iteration index, we record the runtime at each iteration for \dstop, \rnc, \rncc, and \rnccc and plot the results in \cref{fig:round_time}. For VD, the per-iteration time increases steadily, consistent with our analysis in \cref{subsec:vd} that each round enlarges the MILP problem (additional auxiliary variables). In contrast, BD and RBD show fluctuating runtimes with no clear trend: BD changes the coefficients of the accumulated objective each round, and RBD re-samples the reference, both of which alter the search landscape and lead to non-monotonic solver times.

\begin{table}[tbp]
  \centering
  \caption{Comparison of diversification methods}
  \label{tab:diversification-comparison}
  \resizebox{\textwidth}{!}{%
  \begin{threeparttable}
  \setlength{\tabcolsep}{5pt}
\begin{tabular}{>{\raggedright}p{8em} >{\raggedright}p{7em} >{\raggedright}p{7em} >{\raggedright}p{7em} >{\raggedright\arraybackslash}p{7em}}
    \toprule
    Criteria & \textbf{SP} & \textbf{BD} & \textbf{RBD} & \textbf{VD} \\
    \hline
    \textbf{Total runtime} &
      \feature{++}{1--100} &
      \feature{-}{Varying with timeout setting} &
      \feature{-}{Varying with timeout setting} &
      \feature{+}{Varying with timeout setting} \\
      \hline
    \textbf{Stable per-trace time} &
      \feature{--} &
      \feature{+} Varying across iterations  &
      \feature{+} Varying across iterations  &
      \feature{-} Time increase per iteration\\
      \hline
    \textbf{Diversity} &
      \feature{-} &
      \feature{++} &
      \feature{++} &
      \feature{++} \\
      \hline
    \textbf{Compatibility with other tools} &
      \feature{++} Support all MILP-based formulations &
      \feature{+} Compatible only with full Boolean-variable formulations &
      \feature{+} Compatible only with full Boolean-variable formulations &
      \feature{-} Limited compatibility\\
        \specialrule{0.04em}{0pt}{1pt}
      \textbf{Deterministic} &
      \cmark &
      \cmark &
      \xmark &
      \cmark \\
    \bottomrule
  \end{tabular}
    \centering
    \feature{++} = good; \feature{+} = moderate; \feature{-} = poor; \feature{--} = not applicable.
    \end{threeparttable}
    }

\end{table}

\myparagraph{Summary} \Cref{tab:diversification-comparison} contrasts the three diversification strategies. Achieving diversity incurs extra runtime for all of them. The ideas behind BD and RBD are broadly portable: for problems with only Boolean decision variables, adding (possibly randomized) Boolean distance terms can produce multiple diverse solutions. By contrast, VD depends on continuous state variables and absolute-value linearization, which enlarges the MILP problem and makes the approach less directly compatible with other tools or settings.

\section{Conclusion}\label{sec:conclusion}

We introduced diversification into MILP-based STL trace synthesis. We formulated several complementary objectives---\eg Boolean distance, randomized Boolean distance, and value distance---that promote dissimilarity by maximizing an appropriate distance between the candidate trace and previously generated ones. We implemented these objectives in STLts-Div and evaluated effectiveness and efficiency on a suite of benchmarks, showing clear improvements over a solution-pool baseline.

A complete MILP encoding of robust semantics remains open. Such an encoding could enable richer preference objectives and stronger guarantees, including true corner-case synthesis (robustness exactly 0) and additional forms of preferential synthesis.

\bibliographystyle{splncs04}
\bibliography{reference.bib}

@inbook{Maler2004,
  title = {Monitoring Temporal Properties of Continuous Signals},
  ISBN = {9783540302063},
  ISSN = {1611-3349},
  url = {http://dx.doi.org/10.1007/978-3-540-30206-3_12},
  DOI = {10.1007/978-3-540-30206-3_12},
  booktitle = {Formal Techniques,  Modelling and Analysis of Timed and Fault-Tolerant Systems},
  publisher = {Springer Berlin Heidelberg},
  author = {Maler,  Oded and Nickovic,  Dejan},
  year = {2004},
  pages = {152–166}
}

@inbook{Sato2024,
  title = {Optimization-Based Model Checking and Trace Synthesis for Complex STL Specifications},
  ISBN = {9783031656330},
  ISSN = {1611-3349},
  url = {http://dx.doi.org/10.1007/978-3-031-65633-0_13},
  DOI = {10.1007/978-3-031-65633-0_13},
  booktitle = {Computer Aided Verification},
  publisher = {Springer Nature Switzerland},
  author = {Sato,  Sota and An,  Jie and Zhang,  Zhenya and Hasuo,  Ichiro},
  year = {2024},
  pages = {282–306}
}

@article{Fainekos2009,
  title = {Robustness of temporal logic specifications for continuous-time signals},
  volume = {410},
  ISSN = {0304-3975},
  url = {http://dx.doi.org/10.1016/j.tcs.2009.06.021},
  DOI = {10.1016/j.tcs.2009.06.021},
  number = {42},
  journal = {Theoretical Computer Science},
  publisher = {Elsevier BV},
  author = {Fainekos,  Georgios E. and Pappas,  George J.},
  year = {2009},
  month = sep,
  pages = {4262–4291}
}

@misc{gurobi,
  author = {{Gurobi Optimization, LLC}},
  title = {{Gurobi Optimizer Reference Manual}},
  year = 2024,
  note = {\\\url{https://www.gurobi.com}}
}

@misc{appendix,
  author = {Martin Jouve-Genty},
  title = {{Preferential Exemplification for MILP-Based Trace Synthesis of STL Specifications, Online Appendix}},
  year = 2025,
  note = {\\Accessible at \url{https://gitlab.aliens-lyon.fr/mjouvege/stlts-preferential-synthesis}}
}

@inproceedings{lee2021efficient,
  title={Efficient SMT-based model checking for signal temporal logic},
  author={Lee, Jia and Yu, Geunyeol and Bae, Kyungmin},
  booktitle={2021 36th IEEE/ACM international conference on automated software engineering (ASE)},
  pages={343--354},
  year={2021},
  organization={IEEE}
}

@inproceedings{prabhakar2018automatic,
  title={Automatic trace generation for signal temporal logic},
  author={Prabhakar, Pavithra and Lal, Ratan and Kapinski, James},
  booktitle={2018 IEEE Real-Time Systems Symposium (RTSS)},
  pages={208--217},
  year={2018},
  organization={IEEE}
}

@article{alur1996benefits,
  title={The benefits of relaxing punctuality},
  author={Alur, Rajeev and Feder, Tom{\'a}s and Henzinger, Thomas A},
  journal={Journal of the ACM (JACM)},
  volume={43},
  number={1},
  pages={116--146},
  year={1996},
  publisher={ACM New York, NY, USA}
}

@article{rabinovich1998decidability,
  title={On the decidability of continuous time specification formalisms},
  author={Alexander Moshe Rabinovich},
  journal={Journal of logic and computation},
  volume={8},
  number={5},
  pages={669--678},
  year={1998},
  publisher={OUP}
}

@article{bae2019bounded,
  title={Bounded model checking of signal temporal logic properties using syntactic separation},
  author={Bae, Kyungmin and Lee, Jia},
  journal={Proceedings of the ACM on Programming Languages},
  volume={3},
  number={POPL},
  pages={1--30},
  year={2019},
  publisher={ACM New York, NY, USA}
}

@article{gurobi2019gurobi,
  title={Gurobi Optimization, LLC, Gurobi optimizer reference manual},
  author={Gurobi Optimization, LLC},
  journal={URL http://www. gurobi. com},
  year={2019}
}

@inproceedings{duggirala2011abstraction,
  title={Abstraction refinement for stability},
  author={Duggirala, Parasara Sridhar and Mitra, Sayan},
  booktitle={2011 IEEE/ACM Second International Conference on Cyber-Physical Systems},
  pages={22--31},
  year={2011},
  organization={IEEE}
}

@inproceedings{bu2022arch,
  title={Arch-comp22 category report: Hybrid systems with piecewise constant dynamics and bounded model checking},
  author={Bu, Lei and Frehse, Goran and Kundu, Atanu and Ray, Rajarshi and Shi, Yuhui and Zaffanella, Enea and others},
  booktitle={Proceedings of 9th International Workshop on Applied Verification of Continuous and Hybrid Systems (ARCH22). EPiC Series in Computing},
  volume={90},
  pages={44--57},
  year={2022}
}

@inproceedings{reimann2024temporal,
  title={Temporal logic formalisation of ISO 34502 critical scenarios: modular construction with the RSS safety distance},
  author={Reimann, Jesse and Mansion, Nico and Haydon, James and Bray, Benjamin and Chattopadhyay, Agnishom and Sato, Sota and Waga, Masaki and Andr{\'e}, {\'E}tienne and Hasuo, Ichiro and Ueda, Naoki and others},
  booktitle={Proceedings of the 39th ACM/SIGAPP Symposium on Applied Computing},
  pages={186--195},
  year={2024}
}

@inproceedings{cheng2023behavexplor,
  title={Behavexplor: Behavior diversity guided testing for autonomous driving systems},
  author={Cheng, Mingfei and Zhou, Yuan and Xie, Xiaofei},
  booktitle={Proceedings of the 32nd ACM SIGSOFT International Symposium on Software Testing and Analysis},
  pages={488--500},
  year={2023}
}

@inproceedings{ernst2021arch,
  title={ARCH-COMP 2021 Category Report: Falsification with Validation of Results.},
  author={Ernst, Gidon and Arcaini, Paolo and Bennani, Ismail and Chandratre, Aniruddh and Donz{\'e}, Alexandre and Fainekos, Georgios and Frehse, Goran and Gaaloul, Khouloud and Inoue, Jun and Khandait, Tanmay and others},
  booktitle={ARCH@ ADHS},
  pages={133--152},
  year={2021}
}

@inproceedings{atkins2013aerospace,
  title={Aerospace cyber-physical systems education},
  author={Atkins, Ella M and Bradley, Justin M},
  booktitle={AIAA Infotech@ Aerospace (I@ A) Conference},
  pages={4809},
  year={2013}
}

@article{tomlin2000game,
  title={A game theoretic approach to controller design for hybrid systems},
  author={Tomlin, Claire J and Lygeros, John and Sastry, S Shankar},
  journal={Proceedings of the IEEE},
  volume={88},
  number={7},
  pages={949--970},
  year={2000},
  publisher={IEEE}
}

@article{ye2008modelling,
  title={Modelling excitable cells using cycle-linear hybrid automata},
  author={Ye, Pei and Entcheva, Emilia and Smolka, Scott A and Grosu, Radu},
  journal={IET systems biology},
  volume={2},
  number={1},
  pages={24--32},
  year={2008},
  publisher={IET}
}

@article{donze2015blustl,
  title={BluSTL: Controller Synthesis from Signal Temporal Logic Specifications.},
  author={Donz{\'e}, Alexandre and Raman, Vasumathi and Frehse, G and Althoff, M},
  journal={ARCH@ CPSWeek},
  volume={34},
  pages={160--8},
  year={2015}
}

@inproceedings{raman2015reactive,
  title={Reactive synthesis from signal temporal logic specifications},
  author={Raman, Vasumathi and Donz{\'e}, Alexandre and Sadigh, Dorsa and Murray, Richard M and Seshia, Sanjit A},
  booktitle={Proceedings of the 18th international conference on hybrid systems: Computation and control},
  pages={239--248},
  year={2015}
}

@inproceedings{raman2014model,
  title={Model predictive control from signal temporal logic specifications: A case study},
  author={Raman, Vasumathi and Maasoumy, Mehdi and Donz{\'e}, Alexandre},
  booktitle={Proceedings of the 4th ACM SIGBED international workshop on design, modeling, and evaluation of cyber-physical systems},
  pages={52--55},
  year={2014}
}

@inproceedings{su2024switching,
  title={Switching controller synthesis for hybrid systems against STL formulas},
  author={Su, Han and Feng, Shenghua and Zhan, Sinong and Zhan, Naijun},
  booktitle={International Symposium on Formal Methods},
  pages={229--247},
  year={2024},
  organization={Springer}
}

@inproceedings{su2025runtime,
  title={Runtime enforcement of CPS against signal temporal logic},
  author={Su, Han and Shankar, Saumya and Pinisetty, Srinivas and Roop, Partha S and Zhan, Naijun},
  booktitle={Proceedings of the 28th ACM International Conference on Hybrid Systems: Computation and Control},
  pages={1--11},
  year={2025}
}

@inproceedings{zhang2023online,
  title={Online causation monitoring of signal temporal logic},
  author={Zhang, Zhenya and An, Jie and Arcaini, Paolo and Hasuo, Ichiro},
  booktitle={International Conference on Computer Aided Verification},
  pages={62--84},
  year={2023},
  organization={Springer}
}

@article{lindemann2018control,
  title={Control barrier functions for signal temporal logic tasks},
  author={Lindemann, Lars and Dimarogonas, Dimos V},
  journal={IEEE control systems letters},
  volume={3},
  number={1},
  pages={96--101},
  year={2018},
  publisher={IEEE}
}

\appendix
\renewcommand{\theHsection}{appendix.\thesection}
\newpage
\section*{Appendix}
\section{STL semantics}\label[appendix]{app:stl}

\begin{definition}[Boolean semantics of STL]\label{def:stl-sem}
            For a signal \(\sigma\) and an STL formula \(\phi\), the \emph{satisfaction relation} \(\sigma\vDash\phi\) between them is defined as follows:
            \begin{align*}
            & \sigma \vDash p \iff  \pi_p(\sigma(0)) \ge 0, \quad \sigma \vDash \neg p \iff \pi_p(\sigma(0)) < 0, \\
            & \sigma \vDash  \phi_1\wedge\phi_2 \iff \sigma\vDash \phi_1\wedge \sigma \vDash\phi_2, \quad 
            \sigma \vDash  \phi_1\vee\phi_2 \iff \sigma\vDash \phi_1\vee \sigma \vDash\phi_2,\\
            &\sigma \vDash  \phi_1\unt_I\phi_2 \iff \exists t\in I,~\left(\sigma^{t}\vDash \phi_2\wedge \forall t'\in[0,t], \sigma^{t'}\vDash\phi_1\right),\\
            &\sigma \vDash \phi_1 \rels_I \phi_2 \iff \forall t\in I,~ \left(\sigma^{t} \vDash \phi_2 \vee \exists t'\in[0,t], \sigma^{t'}\vDash \phi_1\right).
            \end{align*}
        \end{definition}

        \begin{definition}[Robust semantics of STL]
            For a signal \(\sigma\) and an STL formula \(\phi\), the \emph{robust semantics} of STL returns a quantity \(\sem{\sigma}{\phi}\in\Reals\cup\{+\infty,-\infty\}\) that indicates the satisfaction level of \(\sigma\) to \(\phi\):
                \begin{align*}
                    &\sem{\sigma}{\top}  =  +\infty,\quad \sem{\sigma}{\bot} = -\infty, \quad \sem{\sigma}{p}  =  \pi_p(\sigma(0)), \quad \sem{\sigma}{\neg p } = -\pi_p(\sigma(0)), \\
                    &\sem{\sigma}{\phi_1 \wedge \phi_2 }  =  \min(\sem{\sigma}{\phi_1 }, \sem{\sigma}{\phi_2 }),\quad  
                    \sem{\sigma}{\phi_1 \vee \phi_2 }  =  \max(\sem{\sigma}{\phi_1 }, \sem{\sigma}{\phi_2 }),\\
                    &\sem{\sigma}{\varphi_1\unt_I\varphi_2}  =  \sup_{t\in I}\big(\min\big(\sem{\sigma^t}{\varphi_2},\inf_{t'\in [0,t]}\sem{\sigma^{t'}}{\varphi_1}\big)\big),\\
                    &\sem{\sigma}{\varphi_1\rels_I\varphi_2}  =  \inf_{t\in I}\big(\max\big(\sem{\sigma^t}{\varphi_2},\sup_{t'\in [0,t]}\sem{\sigma^{t'}}{\varphi_1}\big)\big).
            \end{align*}
        \end{definition}

    \section{Proof of propositions}\label[appendix]{appendix:proof}
        \propDiv*
            \begin{proof}
            The samples are i.i.d. and uniformly distributed over a set of size \(2^d\). Hence
            \begin{align*}
                \Pr(\text{all \(m\) traces are distinct})
                &= \prod_{k=1}^{m-1} \Pr \big(\reftrace^{k+1} \notin \{\reftrace^1,\dotsc,\reftrace^{k}\}\big) \\
                &= \prod_{k=1}^{m-1} \left(1 - \frac{k}{2^d}\right)
                = \frac{2^d!}{(2^d-m)!\,2^{md}} \raisebox{-2ex}{.}
            \end{align*}
            Therefore,
            \[
                p \;=\; 1 - \Pr(\text{all \(m\) traces are distinct})
                \;=\; 1 - \frac{2^d!}{(2^d-m)!\,2^{md}} \raisebox{-2ex}{.}
            \]
            
            For the upper bound,
            \[
                p = \Pr(\exists i<j,\, \reftrace^{i} = \reftrace^{j})
                \;\le\; \sum_{i<j}\Pr\big(\reftrace^{i}=\reftrace^{j}\big)
                = \binom{m}{2}\frac{1}{2^d} = \frac{m(m-1)}{2^{d+1}}\raisebox{-2ex}{.}
            \]
        \end{proof}
\end{document}